\documentclass[aps,pra,twocolumn,superscriptaddress,showpacs,floatfix]{revtex4-2}
\long\def\/*#1*/{}
\usepackage{amsfonts}
\pdfoutput=1
\usepackage{graphicx}
\usepackage{amsmath}
\usepackage{amsthm}
\usepackage{comment}
\usepackage[colorlinks=true, citecolor=blue, urlcolor=blue ]{hyperref}
\usepackage{graphicx}

\usepackage{braket}
\usepackage[author={AB}]{pdfcomment}
\usepackage{graphicx}
\usepackage{subcaption}
\usepackage{float}    
\usepackage{caption}
\usepackage{ragged2e} 
\newtheorem{theorem}{Theorem}
\newtheorem{corollary}{Corollary}
\newtheorem{proposition}{Proposition}

\newtheorem{remark}{Remark}
\newtheorem{lemma}{Lemma}

\begin{document}
\title{Quantifying Entangling Power of Controlled Unitary Gates}
\author{Ankur Haldar}
\email{ankurivuhaldar@gmail.com}
\affiliation{Center for Quantum Engineering, Research, and Education, TCG CREST, Bidhan Nagar, Kolkata - 700091, India.}

\author{Pankaj Agrawal}
\email{pankaj.agrawal@tcgcrest.org}
\affiliation{Center for Quantum Engineering, Research, and Education, TCG CREST, Bidhan Nagar, Kolkata - 700091, India.}

\author{Prasenjit Deb}
\email{devprasen@gmail.com}
\affiliation{Center for Quantum Engineering, Research, and Education, TCG CREST, Bidhan Nagar, Kolkata - 700091, India.}

\date{\today}

\begin{abstract}
Applying controlled unitary gates to generate entanglement between qubits is a routine task in both quantum communication and computation. The existing tools for predicting how much entanglement a given gate can generate require either  simulation of the entangling circuit or averaging over a distribution of inputs for a given controlled unitary gate. Here, we introduce a computable quantity $\zeta$ that not only determines whether a controlled unitary gate generates entanglement, but also quantifies the entanglement for any specific input without requiring the construction of the output state. For two-qubit controlled unitary gates, we establish the physical conditions corresponding to the extremum values of the proposed quantity. Extending the dimension of control and target registers to arbitrary size through a generalized controlled unitary architecture, we derive a universal upper bound on $\zeta$ and identify the conditions for its saturation. Later we establish functional relation between the quantity and other known quantities, such as purity, normalized linear entropy, and von Neumann entropy. Finally, we compare our results with previous research work and show that the quantity proposed in this work achieves the previously known optimal values of entangling power of controlled unitary gates for some specific dimensions of target and control registers.
\end{abstract}

\maketitle

\section{Introduction}

Entanglement is the resource that separates quantum information processing tasks from their classical counterparts, and its generation is an essential task in both quantum communication and computation~\cite{Horodecki_2009}. In practice, entanglement generation is almost always accomplished using controlled unitary gates by conditioning an operation on one register upon the state of another~\cite{Nielsen_Chuang_2010}. It is well-known that two-qubit controlled unitary gates are universal for quantum computation when supplemented with arbitrary single-qubit operations~\cite{Barenco_1995,DiVincenzo_1995}. Existing tools for assessing how much entanglement a given controlled unitary gate can generate fall into two broad categories. One is the direct method, which requires simulating the entangling circuit and computing the entanglement of the resultant state. However, the cost incurred in this method grows exponentially with the size of the control and target registers. Another method, which is an indirect one, is based on averaging over a distribution of inputs states for a given controlled unitary operator, signifying the entangling power of that operator~\cite{Zanardi_2000}. The above mentioned method considers ideal scenarios with no noise acting during the controlled operation and the gates are perfect having unit fidelity. In realistic scenarios, however, the gate operations are not perfect due to noise, leading to noisy entangling power \cite{noisy_gates}.

In this article, we consider ideal scenarios and introduce a quantity $\zeta$, which can be considered as an alternative mathematical definition of entangling power, for the class of controlled unitary gates. Given only the probability distribution of the control register's basis states and the pairwise trace overlaps of the target unitaries conditional to these states, the newly proposed quantity not only determines whether a given controlled unitary gate will generate entanglement between the control and target registers but also quantifies the amount of entanglement for any specific input states without requiring simulation of the full state vector. We first mathematically define the proposed computable quantity for a general two-qubit controlled unitary gate, establish the physical conditions corresponding to its extremum values, and prove that the maximum can be achieved only when the trace of the target unitary is zero. Then we introduce generalized controlled unitary (GCU) gates and extend the definition of the quantity for control and target registers of arbitrary dimension. For this multi-qubit setting, we derive a universal upper bound on $\zeta$ and illustrate the bound with the example of Toffoli gate. 
Later we establish functional relation between the newly defined quantity to the conventional measures of entanglement, including purity, linear entropy, and von Neumann entropy. Finally, we connect $\zeta$ to operator entangling power as shown in \cite{Zanardi_2000} by deriving its explicit average over Haar random inputs for the controlled unitary architecture and show that the resulting bound for entangling power is not just tighter than that in the original work, but for several dimension, the quantity defined in our work exactly reproduces the numerically obtained optimal values in the same.

The remainder of the article is arranged as follows: In Section (\ref{sec:2}), we define $\zeta$ for two-qubit controlled gates. We extend the definition of the quantity for control and target registers of arbitrary size in Section (\ref{sec:3}). Then we connect the quantity with the already known measures of entanglement in Section (\ref{sec:4}). In Section (\ref{sec:5}), we compare our approach with previous research work. We summarize the framework's predictive use in circuit design and its relation to existing gate characterization techniques in Section (\ref{sec:6}), before concluding in Section (\ref{sec:7}).

\section{Controlled Unitary Gates in two-qubit systems}
\label{sec:2}

\subsection{CNOT Gate}
We begin with a fundamental bipartite controlled unitary gate architecture comprising of two-qubits, $|q_1\rangle$ being the control and $|q_2\rangle$ the target. Without loss of generality and ignoring the global phase, the state of the qubits can be parameterized as
\begin{eqnarray}
|q_1\rangle &=& a|0\rangle + e^{i\theta}\sqrt{1-a^2}|1\rangle,\nonumber\\
|q_2\rangle &=& b|0\rangle + e^{i\phi}\sqrt{1-b^2}|1\rangle, \label{eq:1}
\end{eqnarray}
where $a, b \in [0,1]$ are the amplitudes, and $\theta, \phi \in [0,2\pi)$ denote the relative phases. Applying a CNOT operation between these qubits, we get a two-qubit state as
\begin{eqnarray}
|\Psi\rangle &=& ab|00\rangle + e^{i\phi}a\sqrt{1-b^2}|01\rangle\nonumber\\
&+& e^{i(\theta+\phi)}\sqrt{(1-a^2)(1-b^2)}|10\rangle\nonumber\\
&+& e^{i\theta}b\sqrt{1-a^2}|11\rangle.
\end{eqnarray}
We notice that by applying a local unitary operation. one can eliminate the phase of the control register, but the phase of the target register remains. Therefore the entanglement of the resulting state depends on the phase in target state. Applying the phase rotation on the $|1\rangle$ of the control qubit we see,
\begin{eqnarray}
|\Psi'\rangle &=& ab|00\rangle + e^{i\phi}a\sqrt{1-b^2}|01\rangle\nonumber\\ 
&+& e^{i\phi}\sqrt{(1-a^2)(1-b^2)}|10\rangle\nonumber\\ 
&+& b\sqrt{1-a^2}|11\rangle. \label{eq:3}
\end{eqnarray}
Tracing out the target register yields the reduced density matrix $\rho_1 = \text{Tr}_2(|\Psi'\rangle\langle\Psi'|)$ of the control qubit. Its eigenvalues are $\lambda_{\pm} = \frac{1 \pm \sqrt{1-4\zeta}}{2}$, where $\zeta \equiv \det(\rho_1 )$. The $\zeta$ is a well known measure of the entanglement of two-qubit systems. Its relation to concurrence $C = 2\sqrt{\zeta}$ \cite{Wootters_1998,Wootters_Hill,Wootters_Coffman_2000}. The explicit evaluation of the determinant of $\rho_1$ gives
\begin{equation}
\zeta = a^2(1-a^2)\left[1 - 4b^2(1-b^2)\cos^2\phi\right]. \label{eq:4}
\end{equation}
For the state to be maximally entangled or separable, the values of this quantity are $\zeta = \frac{1}{4}~ \text{or}~ \zeta = 0$, respectively. Clearly, the resultant two-qubit state will be non-maximally entangled for $0<\zeta<1/4$. We now obtain the conditions on initial states for these values of $\zeta$.

\subsubsection{$\zeta = 0$}
From Eq.~(\ref{eq:4}), we see that this value holds for two physical scenarios -- 
\\
i)  The control qubit completely lacks a coherent superposition, disabling the target's conditional split, i.e., $a \in \{0,1\}$.
\\
ii) The initial target qubit is the eigenstate of the operator $X$, namely, the $|+\rangle$ or $|-\rangle$ states. From Eq. (\ref{eq:4}), the mathematical condition is $4b^2(1-b^2)\cos^2\phi = 1$. This equation has unique solution for $b ~\text{and}~ \phi$. It is $b = \frac{1}{\sqrt{2}}$ and $\phi \in \{0, \pi\}$.

\subsubsection{$\zeta = 1/4$}
The physical scenario leading to this value is -- 
\\
i) When the control qubit is prepared in a coherent superposition of the basis states, i.e., $a = \frac{1}{\sqrt{2}}$, and
\\
ii) the initial target state is prepared in any of the states with $\phi \in \{\frac{\pi}{2}, \frac{3\pi}{2}\}$ regardless of the amplitude parameter $b$, i.e, $4b^2(1-b^2)\cos^2\phi = 0$. This condition implies that initializing the target qubit with a relative phase of $\pm i$ guarantees maximal entanglement yield. 
The target qubit takes the explicit form of,
    \begin{equation}
    |q_2\rangle = b|0\rangle \pm i\sqrt{1-b^2}|1\rangle.
    \end{equation}
For initial $|q_1\rangle = |\pm\rangle$ and $|q_2\rangle = b|0\rangle \pm i\sqrt{1-b^2}|1\rangle$, the resulting state is
    \begin{eqnarray}
    |\Psi\rangle &=& \frac{1}{\sqrt{2}}|0\rangle(b|0\rangle \pm i\sqrt{1-b^2}|1\rangle) \nonumber \\ &\pm& \frac{1}{\sqrt{2}}|1\rangle(b|1\rangle \pm i\sqrt{1-b^2}|0\rangle). \label{eq:maxenttwoqubit}
    \end{eqnarray}

These two-qubit maximally entangled states are local unitary equivalent to the standard Bell states \cite{Horodecki_2009, Nielsen_Chuang_2010}, i.e. the Bell state $|\phi^+\rangle = \frac{1}{\sqrt{2}}[|00\rangle + |11\rangle]$ can be transformed into $|\Psi\rangle$ from Eq. (\ref{eq:maxenttwoqubit}) by applying $I \otimes U$, where
\begin{equation}
U = \begin{pmatrix} b & \pm i \sqrt{1-b^2} \\ \pm i \sqrt{1-b^2} & b \end{pmatrix}
\end{equation}

\subsection{Generalized controlled unitary gates}
We extend the above analysis by replacing the CNOT gate with an arbitrary single-qubit unitary gate $U$ and conditioning its operation by the target qubit's state. We assume that the unitary acts on the target qubit \textit{iff} the control qubit is in the state $|1\rangle$, i.e., $CU = |0\rangle \langle 0| \otimes I + |1\rangle \langle 1| \otimes U$. The initial control and target states are same as that in Eq.(\ref{eq:1}).

Computing the determinant of reduced state of the control qubit from the resultant two-qubit state formed after the action of the  generalized controlled unitary interaction yields 
\begin{equation}
\zeta = a^2(1-a^2)\left[1 - \text{Tr}(U^\dagger\rho_B)\text{Tr}(U\rho_B)\right], \label{eq:6}
\end{equation}
where $\rho_B = |q_2\rangle\langle q_2|$ represents the initial target state. It can be easily verified that taking $U = X$ leads to Eq.(\ref{eq:4}).
If we do the same analysis as done in the previous section, then we can find that the mathematical condition in terms of $\zeta$ is the same as that for a CNOT gate. When $\zeta = 0$, zero entanglement is generated between the two qubits. Such a case can arise when the initial target state $|q_2\rangle$ is an eigenstate of the unitary $U$. In such case, we can write, $U|q_2\rangle = e^{i\theta}|q_2\rangle$ and $\text{Tr}(U^\dagger\rho_B)\text{Tr}(U\rho_B) = |\text{Tr}(U|q_2\rangle \langle q_2|)|^2 = 1$. When $\zeta = 1/4$, a two-qubit maximally entangled entangled state is generated between the control and target qubits. However, the physical condition leading to successful generation of this entanglement is more general.

\begin{theorem}[]
\label{traceless_unitary}
In a two-qubit system, if qubit 1 is prepared in state $|+\rangle$ or $|-\rangle$ and qubit 2 is prepare in state $|q_2\rangle$ and a controlled unitary gate ($|0\rangle\langle0| \otimes I + |1\rangle\langle1| \otimes U$) is applied, then there will always exist a $|q_2\rangle$ that will lead to a maximally entangled state iff the unitary operator $U$ is traceless, i.e.

\begin{equation}
\text{Tr}(U) = 0.\nonumber\\
\end{equation}

\end{theorem}

\begin{proof}
From Eq.(\ref{eq:6}), we find that satisfying the condition of $\zeta = 1/4$ for the generation of the maximal entanglement between the two qubits requires $a = \frac{1}{\sqrt{2}}$, which is fulfilled by the preparation method of the first qubit, and $\text{Tr}(U^\dagger\rho_B) = \text{Tr}(U\rho_B) = 0$. Parameterizing $U$ as a general $2 \times 2$ unitary matrix, we get
\begin{equation}
U = e^{i\alpha}\begin{pmatrix} e^{i\beta_2}\cos\beta_1 & e^{i\beta_3}\sin\beta_1 \\ -e^{-i\beta_3}\sin\beta_1 & e^{-i\beta_2}\cos\beta_1 \end{pmatrix} \label{eq:unitarymatrix}
\end{equation}
and putting this in $\text{Tr}(U^\dagger\rho_B) = \text{Tr}(U\rho_B) = 0$ yields
\begin{eqnarray}
\cos\beta_1\cos\beta_2 = 0,\nonumber\\
(2b\sqrt{1-b^2})\sin\beta_1\sin(\phi+\beta_3) = \nonumber \\ \cos\beta_1\sin\beta_2(1-2b^2). \label{eq:8}
\end{eqnarray}

In the above equation, we have used $\rho_B = |q_2\rangle \langle q_2|$, where $|q_2\rangle = |q_2(b,\phi)\rangle = b|0\rangle + e^{i\phi}\sqrt{1-b^2}|1\rangle$. Considering $\cos \beta_1 = 0$ results to $b \in \{0,1\}$ or a phased condition $\phi = n\pi - \beta_3$. Considering $\cos\beta_1 \neq 0 ~\text{and}~ \cos\beta_2 = 0$, and setting $b = \sin\gamma$ in this equation, we get:
\begin{equation}
\tan\beta_1\sin(\phi+\beta_3) = \cot(2\gamma)\sin\beta_2.
\end{equation}
Because a valid matching pair $(b, \phi)$ can always be computed for any choice of $(\beta_1,\beta_2, \beta_3)$, one can always find a suitable state $|q_2 (b,\phi)\rangle$ to create maximum entanglement between the target and the control given the condition of $\cos \beta_1 \cos \beta_2 = 0$ in Eq.(\ref{eq:8}) i.e. $tr(U)=0$ is satisfied. 

Conversely, suppose $|q_2\rangle$ leads to maximal entanglement, so that $\mathrm{Tr}(U\rho_B)=0$. Multiplying out Eq.~(\ref{eq:unitarymatrix}) against $\rho_B=|q_2\rangle\langle q_2|$ for general $b,\phi$ gives,
\begin{eqnarray}
\mathrm{Tr}(U\rho_B) &=& e^{i\alpha}\Big[\cos\beta_1\cos\beta_2 \nonumber \\
&+& i\big(\cos\beta_1(2b^2-1)\sin\beta_2 \nonumber \\
&+& 2b\sqrt{1-b^2}\sin\beta_1\sin(\phi+\beta_3)\big)\Big],
\end{eqnarray}
whose real part, $\cos\beta_1\cos\beta_2$, is independent of $b$ and $\phi$ and it is precisely the first line of Eq.~(\ref{eq:8}), which therefore holds for every target state. Since $\mathrm{Tr}(U) = 2e^{i\alpha}\cos\beta_1\cos\beta_2$ by Eq.~(\ref{eq:unitarymatrix}), this gives
\begin{equation}
\mathrm{Re}\big[e^{-i\alpha}\mathrm{Tr}(U\rho_B)\big] = \cos\beta_1\cos\beta_2 = \frac{\mathrm{Tr}(U)}{2e^{i\alpha}}
\end{equation}
for every $(b,\phi)$. If $\mathrm{Tr}(U)\ne0$, this real part is a fixed nonzero number regardless of the target state, so $\mathrm{Tr}(U\rho_B)\ne0$ for every $|\psi\rangle$, and maximal entanglement is unreachable. Hence $\mathrm{Tr}(U\rho_B)=0$ for some $|q_2\rangle$ forces $\mathrm{Tr}(U)=0$.

\end{proof}

To complete the analysis of controlled unitary gates in two-qubit systems, we relax the requirement that the target qubit experiences an identity operation when the control qubit is in $|0\rangle$ state. We introduce Generalised Controlled Unitary (GCU) gate, which functions in such a way that two unitary operators $U_0$ and $U_1$ act on the target when the control is in $|0\rangle$ and $|1\rangle$ states, respectively, i.e., 
\begin{equation}
GCU = |0\rangle\langle0| \otimes U_0 + |1\rangle\langle1| \otimes U_1.
\end{equation}
Applying this gate between the control and target qubits and tracing out the target system yields,
\begin{equation}
\zeta = a^2(1-a^2)\left[1 - \text{Tr}(U_0\rho_B U_1^\dagger)\text{Tr}(U_1\rho_B U_0^\dagger)\right], \label{eq:zetafortwoqubit}
\end{equation}
where $\rho_B = |q_2\rangle \langle q_2|$. 

\begin{theorem}[Relative Unitary Theorem]
\label{traceless_unitary}
 A Generalised Controlled Unitary gate utilizing the operator pair $(U_0, U_1)$ creates the exact same quantitative entanglement as a standard controlled unitary gate using the relative unitary operator $U_{rel} = U_1^\dagger U_0$ or $U_{rel} = U_0^\dagger U_1$.
\end{theorem}

\begin{proof}
By exploiting the cyclic permutation property of the trace operator i.e. $\text{Tr}(AB)=\text{Tr}(BA)$, we rewrite Eq.~(\ref{eq:zetafortwoqubit}) as:
\begin{equation}
\zeta = a^2(1-a^2)\left[1 - \text{Tr}(U_1^\dagger U_0\rho_B)\text{Tr}((U_1^\dagger U_0)^\dagger\rho_B)\right].
\end{equation}
Comparing with Eq.~(\ref{eq:6}), we can say that this is the same entanglement created by the operator $|0\rangle \langle 0| \otimes I + |1\rangle \langle 1| \otimes U_{rel}$, which is
\begin{equation}
\zeta = a^2(1-a^2)\left[1 - \text{Tr}(U_{rel}\rho_B)\text{Tr}(U_{rel}^\dagger\rho_B)\right],
\end{equation}
where $U_{rel}$ is either $U_1^\dagger U_0$ or $U_0^\dagger U_1$.
\end{proof}

An immediate consequence of Theorem 2 is that if the relative operator $U_1^\dagger U_0$ is traceless ($\text{Tr}(U_1^\dagger U_0) = 0$), the GCU gate can generate a maximally entangled state between the qubits given optimal input states. Furthermore, this entanglement is invariant under the swapping $U_0 \leftrightarrow U_1$ i.e.
\begin{equation}
GCU_{\text{swapped}} = |0\rangle\langle0| \otimes U_1 + |1\rangle\langle1| \otimes U_0.
\end{equation}

\section{Controlled unitary gates in Multi-Qubit Systems}
\label{sec:3}

\subsection{Multi-qubit GCU architecture}
We extend our analysis of GCU gates to multi-qubit systems, where the control and target registers contain $m$ and $n$ qubits, respectively. The registers are initialized in the pure states $\rho_c$ and $\rho_t$. To construct a GCU gate that will act between the registers, we require $2^m$ target unitary operators $\{U_i\}$, each of dimension $2^n \times 2^n$:
\begin{equation}
GCU = \sum_{i=0}^{2^m-1} |i\rangle\langle i| \otimes U_i,
\end{equation}
where $\{|i\rangle\}$ represents the set of computational basis states of the control register. To get the usual multi qubit controlled unitary gate, all the $\{ U_i \}$ except the desired one, need to be set to identity. In \cite{Zanardi_2000}, the authors discuss GCU gates but for equal dimensional target and control registers. The reduced density matrix of the control register after the gate operation can be written as,
\begin{equation}
\rho_c^{\mathrm{final}} = \sum_{i,j=0}^{2^m-1} |i\rangle\langle j| \rho_c^{ij} \text{Tr}(U_i \rho_t U_j^\dagger), \label{eq:15}
\end{equation}
where the terms $\rho_c^{ij} = \langle i|\rho_c|j\rangle$ signify the coherence of the initial states of the control register.  

\subsection{Generalization of the function $\zeta$}

From the expression of $\zeta$ for two qubit GCU case, we can generalize the function from Eq.~(\ref{eq:zetafortwoqubit}) to multi qubit systems as:
\begin{equation}
\zeta = \sum_{i<j} p_i p_j \left[1 - \text{Tr}(U_i \rho_t U_j^\dagger)\text{Tr}(U_j \rho_t U_i^\dagger)\right], \label{eq:16}
\end{equation}
where $p_i = \rho_c^{ii}$ represents the probability of finding the initial control register in the state $|i\rangle$.
This equation represents a weighted sum over the unitary distinguishability of all unique pairs. A spatial permutation of the gate elements $\{U_i\}$ that keeps $\zeta$ invariant does not affect the entanglement between the target and control registers. If the control register is prepared in a uniform superposition ($p_i = \frac{1}{2^m} \, \forall i$) of the basis states, the entanglement yield is completely invariant under any spatial permutation of the gate elements $\{U_i\}$ within the circuit architecture.

We express the purity of the reduced control state directly
in terms of $\zeta$. Since the control register is prepared in a pure
state, $\rho_c^{ij}=c_ic_j^*$, with $\{c_i\}$ being the amplitude of the ith state, so from Eq.(\ref{eq:16}), $[\rho_c^{\mathrm{final}}]_{ij}=c_ic_j^*\Lambda_{ij}$ with
$\Lambda_{ij}=\mathrm{Tr}(U_i\rho_tU_j^\dagger)$. Noting that
$\Lambda_{ii}=\mathrm{Tr}(U_i\rho_tU_i^\dagger)=\mathrm{Tr}(\rho_t)=1$ and
$\Lambda_{ji}=\Lambda_{ij}^*$, a direct
computation gives,
\begin{align}
\gamma = \mathrm{Tr}\!\left[(\rho_c^{\mathrm{final}})^2\right]
&= \sum_{i,j}p_ip_j\,\Lambda_{ij}\Lambda_{ji} \nonumber\\
&= \sum_i p_i^2 + 2\sum_{i<j}p_ip_j|\Lambda_{ij}|^2 .
\end{align}
Using $\left(\sum_ip_i\right)^2=1$ and $\sum_ip_i^2=1-2\sum_{i<j}p_ip_j$, we finally get,
\begin{equation}
\gamma = 1-2\sum_{i<j}p_ip_j\left(1-|\Lambda_{ij}|^2\right) = 1-2\zeta .
\label{eq:purity-zeta}
\end{equation}
This identity is used again in Section (\ref{sec:4}) while connecting $\zeta$ to the established measures of entanglement.

\begin{remark}
This computation shows that $\gamma$, and hence $\zeta$, depends on the initial state of the control register only through the populations $\{p_i\}$, and not through any relative phase between the
control amplitudes $c_i=\sqrt{p_i}\,e^{i\theta_i}$. Indeed, $\rho_c^{ij}\rho_c^{ji}=(c_ic_j^*)(c_jc_i^*)=p_ip_j\,e^{i(\theta_i-\theta_j)}e^{i(\theta_j-\theta_i)}=p_ip_j$ identically, for any choice of $\theta_i,\theta_j$. The entanglement generated by a GCU gate is therefore fully determined by the probability distribution of the control register's basis states, independent of the relative phases between them beyond their magnitudes. This fact is consistent with $\zeta$'s dependence only in terms of $\{p_i\}$ and not in terms of the control phases $\theta_i$.
\end{remark}

We now bound $\gamma$, and hence $\zeta$, using two standard facts.

\textit{Schmidt rank bound:} Let $|\Psi\rangle\in\mathcal H_c\otimes\mathcal H_t$ be a pure state with
$\dim\mathcal H_c=2^m$ and $\dim\mathcal H_t=2^n$. Then the reduced state
$\rho_c=\mathrm{Tr}_t(|\Psi\rangle\langle\Psi|)$ has rank at most $2^d$,
where $d=\min(m,n)$. \cite{Nielsen_Chuang_2010}

\textit{Minimum purity at fixed rank:} If $\rho$ is a density matrix of rank $r$ with eigenvalues
$\lambda_1,\dots,\lambda_r>0$ summing to $1$, then
$\mathrm{Tr}(\rho^2)\ge 1/r$, with equality iff $\lambda_k=1/r$ for all
$k$. \cite{Nielsen_Chuang_2010}

\begin{theorem}[Universal upper bound on $\zeta$]
\label{thm:zetamax}
For a GCU gate acting between an $m$-qubit control register and an
$n$-qubit target register, both prepared in pure states,
\begin{equation}
0\;\le\;\zeta \;\le\; \zeta_{\max} := \frac{2^d-1}{2^{d+1}}, \qquad d=\min(m,n),
\label{eq:zetamax}
\end{equation}
and this bound is tight.
\end{theorem}

\begin{proof}
By Schmidt rank bound and minimum purity at fixed rank, $\rho_c^{\mathrm{final}}$ has rank
$r\le2^d$, and for fixed rank $r$ the purity
$\gamma$ attains its minimum value $1/r$ when the $r$ nonzero
eigenvalues are equal. Since $1/r$ is decreasing in $r$, the global
minimum over all admissible ranks occurs at the largest allowed rank,
$r=2^d$, giving $\gamma_{\min}=2^{-d}$. By Eq.~(\ref{eq:purity-zeta}),
$\zeta=\tfrac12(1-\gamma)$ is maximized exactly when $\gamma$ is
minimized, so
\begin{equation}
\zeta_{\max} = \frac{1-\gamma_{\min}}{2} = \frac{1-2^{-d}}{2}
= \frac{2^d-1}{2^{d+1}} .
\end{equation}

\emph{Tightness.} Fix any target state $|\psi_t\rangle$. Since
$2^d\le2^n=\dim\mathcal H_t$, choose an orthonormal set
$\{|f_k\rangle\}_{k=1}^{2^d}\subset\mathcal H_t$ and, for each
$k=1,\dots,2^d$, a unitary $U_k$ satisfying $U_k|\psi_t\rangle=|f_k\rangle$;
such a unitary always exists because the unitary group acts transitively
on unit vectors. Assign control populations $p_k=2^{-d}$ for
$k=1,\dots,2^d$, with any remaining control branches unpopulated. Then
$\Lambda_{kl}=\langle f_l|f_k\rangle=\delta_{kl}$ for $k,l\le2^d$, so
\begin{equation}
\zeta = \sum_{k<l\le2^d}2^{-d}\cdot2^{-d}\cdot1
= \binom{2^d}{2}2^{-2d} = \frac{2^d-1}{2^{d+1}} ,
\end{equation}
saturating Eq.~(\ref{eq:zetamax}).
\end{proof}

Saturating the bound only requires that $2^d$ control branches be populated with equiprobability $2^{-d}$, mapping the target state to mutually orthogonal images; when $m>n$, achieving $\zeta_{\max}$ does
\emph{not} require a uniform superposition over all $2^m$ control basis states.

A uniform superposition over \emph{all} $2^m$ control states instead
requires the stronger condition
$\mathrm{Tr}(U_i^\dagger U_j)=0$ for all $i\ne j$. This condition
is considerably more restrictive. To satisfy this, any set of $2^m$ pairwise
Hilbert Schmidt orthogonal unitaries can contain at most one identity
operator, since $\mathrm{Tr}(I^\dagger I)=2^n\ne0$. Conventional multi controlled gates with a single nontrivial unitary
conditioned on one control branch and identity on all others,
generically violate this and therefore cannot saturate the bound under a
uniform control distribution, although they may still saturate it under
a suitably chosen non-uniform one.

When the control and target registers are dimensionally asymmetric, i.e., $dim(\mathcal{H}_c) \neq dim(\mathcal{H}_t)$, where $\mathcal{H}_c$ and $\mathcal{H}_t$ are the Hilbert space of control and target registers respectively, the maximum achievable entanglement is strictly bounded by the dimension of the smaller subsystem. As $d \rightarrow \infty$, the maximum value of the function asymptotically approaches $\lim_{d\rightarrow\infty}\zeta_{max} = 1/2$. This allows us to define a normalized ratio $\zeta_r$, bounded between 0 and 1, as:
\begin{equation}
\zeta_r = \frac{\zeta}{\zeta_{max}} = \frac{2^{d+1}}{2^d - 1}\zeta .
\end{equation}

\subsection{Example: The Toffoli Gate}

As a concrete illustration of Theorem 3 and its achievability
construction, consider the Toffoli gate on a 2-qubit control register
and a 1-qubit target: $U_{00}=U_{01}=U_{10}=I$, $U_{11}=X$, so the
target qubit flips iff both control qubits are $|1\rangle$. Since
$\mathrm{Tr}(I\rho_tI^\dagger)=\mathrm{Tr}(\rho_t)=1$, every pair drawn
from $\{00,01,10\}$ contributes nothing to Eq.~(15); only pairs
involving the $|11\rangle$ branch survive, giving
\begin{equation}
\zeta(p_{11},b,\phi) = p_{11}(1-p_{11})\left[1-4b^2(1-b^2)\cos^2\phi\right],
\label{eq:toffoli-zeta}
\end{equation}
where $p_{11}$ is the probability of the initial control state being in $|11\rangle$, $b,\phi$ parameterize the target state as in Eq.~(\ref{eq:1}). This is the same
functional form as the CNOT result Eq.~(\ref{eq:4}), with $a^2(1-a^2)$ replaced
by $p_{11}(1-p_{11})$.

Here $d=\min(m,n)=\min(2,1)=1$, so by Theorem 3, to achieve maximum entanglement between the two registers, initial control state requires only $2^d=2$ populated control branches, each
with weight $2^{-d}=1/2$: we set $p_{11}=1/2$ (with the remaining weight distributed arbitrarily over the other three branches, which do not affect $\zeta$), giving $\zeta_{\max}=1/4$ and hence, writing
$u:=b^2\in[0,1]$,
\begin{equation}
\zeta_r(u,\phi) = \frac{\zeta}{\zeta_{\max}} = 1-4u(1-u)\cos^2\phi ,
\label{eq:toffoli-zetar}
\end{equation}
bounded in $[0,1]$ as in Eq.~(21).

Fig.~\ref{fig:toffoli-heatmap} displays $\zeta_r(u,\phi)$ in polar
form, with radial coordinate $u=b^2$ and angular coordinate $\phi$, so
that each point sits at Cartesian position $(b^2\cos\phi,\,b^2\sin\phi)$.
Two isolated minima, $\zeta_r=0$, occur at $u=1/2$ ($b=1/\sqrt2$),
$\phi\in\{0,\pi\}$ -- the target prepared in an eigenstate of $X$ --
appearing as two symmetric blue lobes,
while $\zeta_r$ reaches its maximum of $1$ on the outer boundary
$u=1$, at the center $u=0$, and along the vertical axis
$\phi\in\{\pi/2,3\pi/2\}$ for any $u$. The resulting two-fold
symmetry of the pattern directly reflects the $\cos^2\phi$ dependence
in Eq.~(\ref{eq:toffoli-zetar}).

\begin{figure}[htbp]
\centering
\includegraphics[width=\columnwidth]{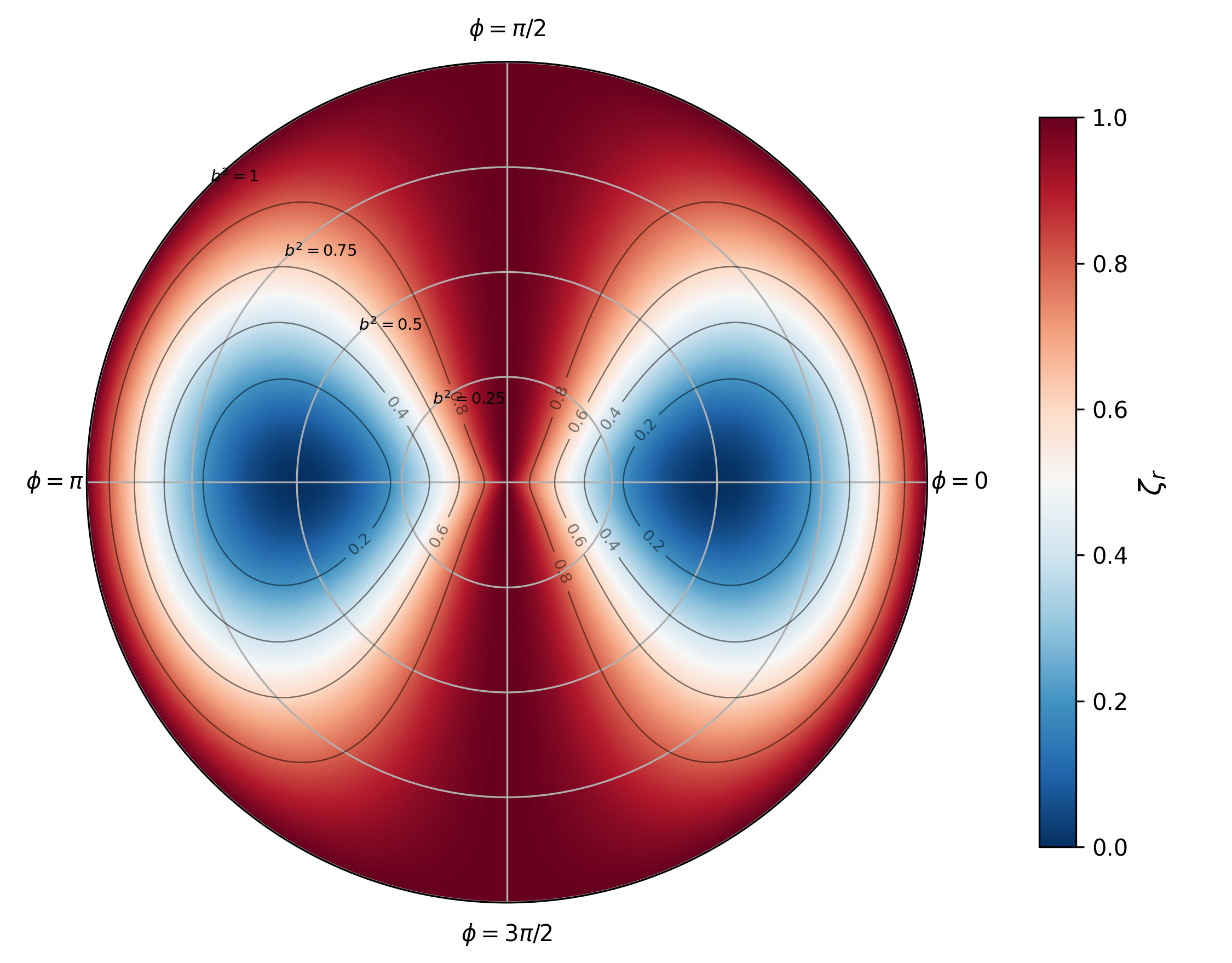}
\caption{\justifying{$\zeta_r(b^2,\phi)$ for the Toffoli gate, Eq.~(\ref{eq:toffoli-zetar}),
at fixed $p_{11}=1/2$, plotted in polar coordinates with radius $b^2$
and angle $\phi$. Blue regions indicate separable control-target
output ($\zeta_r\to0$); red regions indicate maximal entanglement
($\zeta_r=1$). Contour lines are drawn at $\zeta_r=0.2,0.4,0.6,0.8$.}}
\label{fig:toffoli-heatmap}
\end{figure}

\section{Quantification of entanglement}
\label{sec:4}
To validate $\zeta$ as a legitimate metric for quantifying the entanglement generated between the control and target registers after the action of a controlled unitary gate, we connect it to established entanglement measures. 

\textbf{Purity($\gamma$):} Eq. (\ref{eq:purity-zeta}) gives the analytical relation between $\zeta$ and purity ($\gamma$):

\begin{equation}
\gamma = 1 - 2\zeta. \label{eq:21}
\end{equation}

\textbf{Linear Entropy ($S_L$):} The standard definition of normalized linear entropy ($S_L$) for quantum states spanning a $2^d$-dimensional Hilbert space is given by $S_L = \frac{2^d}{2^d - 1}(1 - \gamma)$ \cite{linearentropy1,linearentropy2}. Substituting the expression of purity as derived in Eq.(\ref{eq:purity-zeta}) into this definition yields:
\begin{equation}
S_L = \frac{2^d}{2^d - 1}(1 - (1 - 2\zeta)) = \frac{2^{d+1}}{2^d - 1}\zeta = \zeta_r.
\end{equation}
Clearly, this equation indicates that the ratio $\zeta_r$ derived by us signifies the normalized linear entropy of reduced state of the control register. 

\textbf{Concurrence:} For two-qubit pure entangled states, the function $\zeta$ is the determinant of the reduced density matrix of the control qubit and is related to concurrence $C$ as \cite{Wootters_1998,Wootters_Hill,Wootters_Coffman_2000}
\begin{equation}
C = 2\sqrt{\zeta}.
\end{equation}

\textbf{Von Neumann entropy:} The von Neumann entropy of the reduced state of the control register is a measure of entanglement between the two subsystems. For two-qubit systems with one control and one target qubit, the relation between von Neumann entropy and $\zeta$ is the following, 

\begin{equation}
S(\rho) = -[\lambda_+ \log(\lambda_+) + \lambda_- \log(\lambda_-)],
\end{equation}
where $\lambda_\pm = \frac{1\pm \sqrt{1-4\zeta}}{2}$. 
However, for multi-qubit registers, the relation is not anymore straight forward as $\zeta$ neither is a simple determinant nor directly connected with the eigenvalues of the reduced state of the target register. Nonetheless, $\zeta$ can be related to the von Neumann entropy $S(\rho)$ up to an approximation. Expanding the expression for the entropy, we get, 

\begin{eqnarray}
S(\rho) &=& -\rho~\mbox{ln}\rho\nonumber\\ 
&=&\sum_{i=1}^{\infty} \sum_{k=0}^{i} (-1)^k \frac{(i-1)!}{k!(i-k)!}\text{Tr}(\rho^{k+1}). \label{entropyexpansion}
\end{eqnarray}
For the series to converge to the desired value, the sum over $k$ needs to be calculated before the sum over $i$ \cite{rudin1976, boas2005mathematical}. Writing down the series up to second order gives,

\begin{equation}
    S(\rho) \approx (1-\text{Tr}(\rho^2)) + (\frac{1}{2} - \text{Tr}(\rho^2) + \frac{1}{2} \text{Tr}(\rho^3)) + ... \label{eq:26}
\end{equation}
As $\zeta$ is directly related with the purity of the reduced density matrix, we can easily relate it to an approximate $S(\rho)$. However, to get a more accurate approximation we will need to calculate the trace of higher power of the reduced density matrix. In the same setup that is used throughout the paper, the general analytic form of $\text{Tr}(\rho^k)$ for any $k$ looks like,

\begin{equation}
    \text{Tr}(\rho^k) = \sum_{i_1,..,i_k} p_{i_1}p_{i_2}..p_{i_k} \Lambda_{i_1i_2}\Lambda_{i_2i_3}..\Lambda_{i_ki_1}, \label{eq:tracepowers}
\end{equation}
where $p_i = \rho_c^{ii}$, i.e., the probability of the control register to be in the $i^{th}$ state initially and $\Lambda_{ij}=\text{Tr}(U_i \rho_t U_j^\dagger)$ are the unitary distinguishability of the GCU gate with respect to the target state. Using this equation, the von Neumann entropy of one of the registers can be approximated up to any order in the GCU setup. 

\textbf{Approximation of von Neumann Entropy in terms of $\zeta$:} For a control register consisting of a single qubit (Number of eigenvalues of the density matrix of the control register $=M=2$), the state $\rho_c^{\mathrm{final}}$ has exactly two eigenvalues
$\lambda_+,\lambda_-$, and in this special case, the von Neumann entropy can be approximated in terms of $\zeta$ only. Recall from Eq.~(\ref{eq:purity-zeta}) that $\zeta$ is exactly the second elementary symmetric polynomial of these eigenvalues. For a set of eigenvalues $\{\lambda_1,\dots,\lambda_M\}$, the $k$-th elementary symmetric polynomial is defined as
\begin{equation}
h_k := \sum_{i_1<\cdots<i_k}\lambda_{i_1}\cdots\lambda_{i_k},
\qquad h_0:=1,
\end{equation}
The sum of all products of $k$ distinct eigenvalues, in particular $h_1=\mathrm{Tr}\,\rho$ and $h_2=\sum_{i<j}\lambda_i\lambda_j$. In this notation,
\begin{equation}
\zeta = h_2(\rho),
\end{equation}
since $h_2=\tfrac12\big[(\mathrm{Tr}\,\rho)^2-\mathrm{Tr}(\rho^2)\big]=\tfrac12(1-\gamma)=\zeta$.

Newton's identities~\cite{newton} relate the power sums
$\mathrm{Tr}(\rho^k)$ of a matrix's eigenvalues to its elementary
symmetric polynomials $h_k$:
\begin{equation}
\mathrm{Tr}(\rho^k) = \sum_{i=1}^{k-1}(-1)^{i-1}h_i\,\mathrm{Tr}(\rho^{k-i})
+ (-1)^{k-1}kh_k, ~~ k\le M,
\end{equation}
with no extra final term when $k>M$. For $M=2$ there are only two
eigenvalues, so every elementary symmetric polynomial of order three
or higher vanishes identically, $h_3=h_4=\cdots=0$. Newton's
identities therefore collapse, for all $k\ge2$, to the two-term
recursion

\begin{equation}
\mathrm{Tr}(\rho^k) = \mathrm{Tr}(\rho^{k-1}) - \zeta\,\mathrm{Tr}(\rho^{k-2}),
\ \ \mathrm{Tr}(\rho^0)=2,~ \mathrm{Tr}(\rho^1)=1 .
\label{eq:newton-recursion-M2}
\end{equation}
Eq.~(\ref{eq:newton-recursion-M2}) gives every trace power
$\mathrm{Tr}(\rho^k)$ as an explicit polynomial in $\zeta$ alone, to any
order, without evaluating the multi-index sum of
Eq.~(\ref{eq:tracepowers}) or simulating the control-target state. The
first few are
\begin{align}
\mathrm{Tr}(\rho^2) &= 1-2\zeta, \nonumber\\
\mathrm{Tr}(\rho^3) &= 1-3\zeta, \nonumber\\
\mathrm{Tr}(\rho^4) &= 1-4\zeta+2\zeta^2, \nonumber\\
\mathrm{Tr}(\rho^5) &= 1-5\zeta+5\zeta^2 ,
\end{align}
with higher orders generated directly from Eq.~(\ref{eq:newton-recursion-M2}).

Substituting these into the entropy series of Eq.~(\ref{entropyexpansion}),
truncated at order $N$,
\begin{equation}
S_N(\rho) = \sum_{n=1}^{N}\sum_{k=0}^{n}(-1)^k\frac{(n-1)!}{k!(n-k)!}\,
\mathrm{Tr}(\rho^{k+1}) ,
\label{eq:entropy-truncated}
\end{equation}
now yields $S_N$ as an explicit closed-form polynomial in $\zeta$ at any
chosen order. In particular, truncating at $N=2$ and retaining the exact
relation $\mathrm{Tr}(\rho^3)=1-3\zeta$ gives,
\begin{equation}
S_2(\rho) = \frac{5}{2}\zeta.
\label{eq:S2-corrected}
\end{equation}

Fig.~\ref{fig:cnot-entropy} compares the exact von Neumann entropy
\begin{equation}
S(\rho) = -\lambda_+\log_2\lambda_+ - \lambda_-\log_2\lambda_- ,
~~ \lambda_\pm=\frac{1\pm\sqrt{1-4\zeta}}{2},
\end{equation}
reported in bits so that $S\in[0,1]$ over the full physical range
$\zeta\in[0,\zeta_{\max}]=[0,1/4]$, against the truncated series of
Eq.~(\ref{eq:entropy-truncated}) at orders $N=5$ and $N=10$, evaluated
using the Newton's-identity recursion of Eq.~(\ref{eq:newton-recursion-M2}).
The truncated series converges monotonically to the exact curve as $N$
increases, with the largest residual error at intermediate $\zeta$; both the truncation leads to the exact values at the endpoints $\zeta=0$ and
$\zeta=\zeta_{\max}$.

\begin{figure}[h]
\centering
\includegraphics[width=\columnwidth]{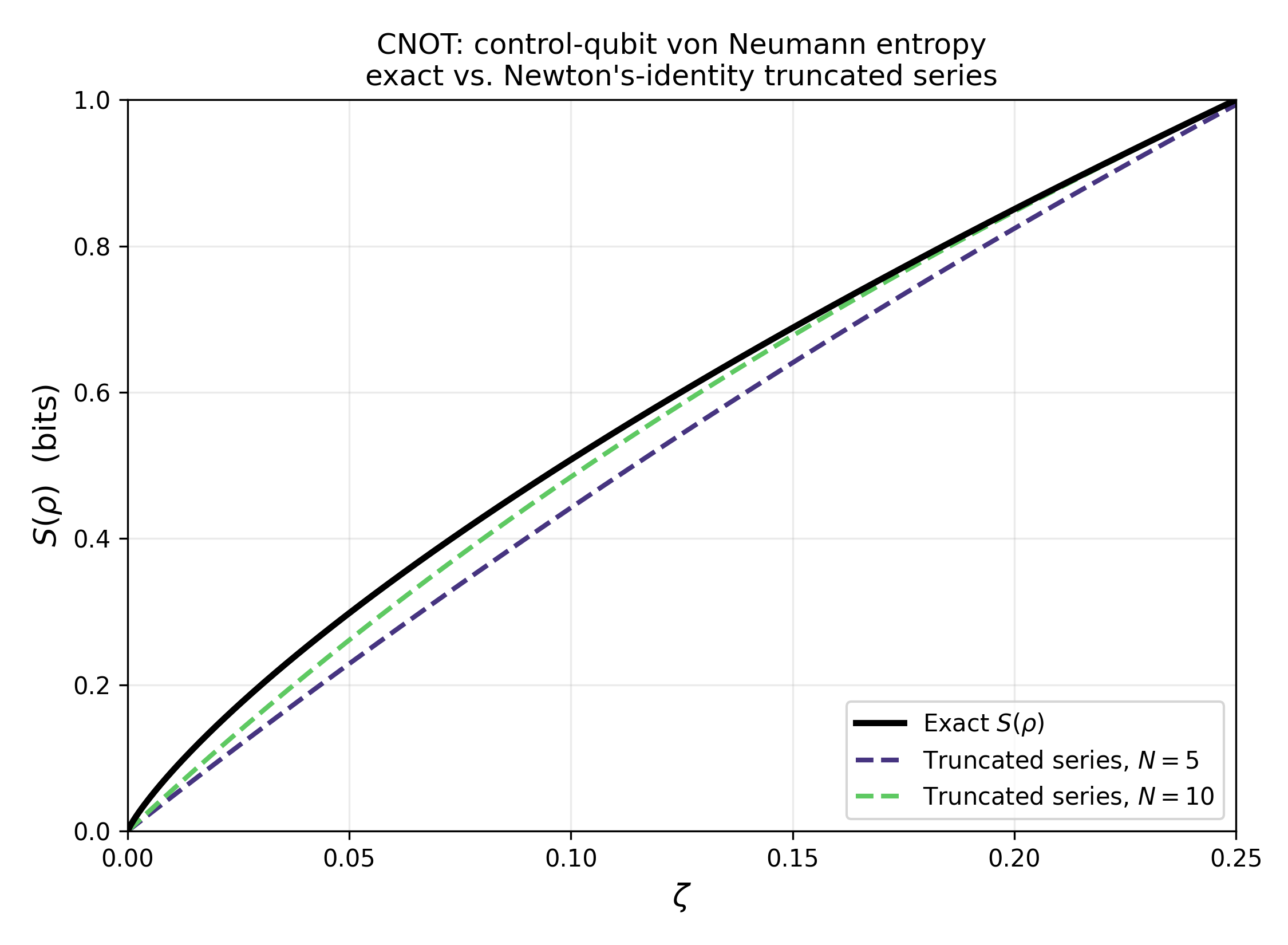}
\caption{\justifying{von Neumann entropy of the control-qubit reduced state for a
CNOT gate, as a function of $\zeta$. The exact entropy (solid black) is
compared against the truncated series of Eq.~(\ref{eq:entropy-truncated}),
evaluated via the Newton's-identity recursion of
Eq.~(\ref{eq:newton-recursion-M2}), at truncation orders $N=5$ and
$N=10$ (dashed). Entropy is reported in bits, so the physical range is
$S\in[0,1]$ over $\zeta\in[0,1/4]$.}}
\label{fig:cnot-entropy}
\end{figure}

\section{Relation to operator entangling power}
\label{sec:5}

Zanardi \textit{et al}~\cite{Zanardi_2000} introduced the entangling
power of a bipartite unitary $U$ acting on $\mathcal H=\mathcal
H_1\otimes\mathcal H_2$ ($\dim\mathcal H_1=d_1$, $\dim\mathcal H_2=d_2$)
as the average, over a distribution $p$ of product input states, of the
linear entropy produced:
\begin{eqnarray}
e_p(U) &:=& \overline{E(U|\psi_1\rangle\otimes|\psi_2\rangle)}^{\,\psi_1,\psi_2}, \nonumber \\
E(|\Psi\rangle) &:=& 1-\mathrm{Tr}_1(\rho^2),\ \ \rho:=\mathrm{Tr}_2|\Psi\rangle\langle\Psi| .
\label{eq:zanardi-def}
\end{eqnarray}
Identifying subsystem 1 with our control register and subsystem 2 with
the target register, $\rho$ in Eq.~(\ref{eq:zanardi-def}) is exactly our
$\rho_c^{\mathrm{final}}$ from Eq.~(\ref{eq:15}), and Eq.~(\ref{eq:purity-zeta}) gives
\begin{equation}
E(|\Psi\rangle) = 1-\gamma = 2\zeta .
\label{eq:E-is-2zeta}
\end{equation}
This holds for any pure product input and any GCU form unitary, and identifies Zanardi's linear entropy as exactly twice our $\zeta$.
Consequently
\begin{equation}
e_p(U) = 2\,\overline{\zeta}^{\,p} , \qquad e_{p_0}(U) = 2\,\overline{\zeta}^{\,p_0} ,
\end{equation}
where $p_0$ denotes the uniform (Haar) distribution over product states, and $\zeta$ on the right is evaluated at the sampled input. Explicit closed form for $\zeta$ in terms of
populations and target unitary overlaps allows this average to be evaluated directly for the GCU subclass, rather than through Zanardi's formalism.

\begin{lemma}
For normalized $|\psi\rangle$, two Haar averaged quantity for indices $i\ne j$, and for a unitary $V$ with dimension $D\times D$ are
\begin{eqnarray}
\overline{|\langle i|\psi\rangle|^2|\langle j|\psi\rangle|^2}
&=& \frac{1}{D(D+1)}, \nonumber \\
\overline{|\langle\psi|V|\psi\rangle|^2} &=& \frac{D+|\mathrm{Tr}\,V|^2}{D(D+1)} .
\label{eq:haar-moments}
\end{eqnarray}
\end{lemma}

\begin{proof}
For probabilities $x_k=|\langle k|\psi\rangle|^2$, two standard identities are used, $\overline{x_ix_j}=1/(D(D+1))$ ($i\ne j$) and $\overline{x_i^2}=2/(D(D+1))$ \cite{Mele_2024haar,geo_q_state,Choi_2023}. The first identity in
Eq.~(\ref{eq:haar-moments}) is immediate. For the second, diagonalize
$V=\sum_k e^{i\theta_k}|k\rangle\langle k|$ and write
$\langle\psi|V|\psi\rangle=\sum_kx_ke^{i\theta_k}$ with
$x_k=|\langle k|\psi\rangle|^2$ in $V$'s eigenbasis. Then
$\overline{|\langle\psi|V|\psi\rangle|^2}
=\sum_k\overline{x_k^2}+\sum_{k\ne l}\overline{x_kx_l}\,e^{i(\theta_k-\theta_l)}
=\frac{2}{D+1}+\frac{|\mathrm{Tr}\,V|^2-D}{D(D+1)}
=\frac{D+|\mathrm{Tr}\,V|^2}{D(D+1)}$,
using $\sum_{k\ne l}e^{i(\theta_k-\theta_l)}=|\mathrm{Tr}\,V|^2-D$.
\end{proof}

\begin{proposition}
\label{prop:ep0-gcu}
For a GCU gate with $M$ control branches (populations $p_i$, target
unitaries $U_i$, $i=1,\dots,M$) and $N$-dimensional target,
\begin{equation}
e_{p_0}(U) = \frac{2}{M(M+1)N(N+1)}\sum_{i<j}\left[N^2-|\mathrm{Tr}(U_j^\dagger U_i)|^2\right] .
\label{eq:ep0-gcu-exact}
\end{equation}
\end{proposition}
\begin{proof}
Since the control and target inputs are independently Haar-random, $\overline{\zeta}^{\,p_0}=\sum_{i<j}\overline{p_ip_j}\cdot\overline{1-|\Lambda_{ij}|^2}$. By Eq.~(\ref{eq:haar-moments}), $\overline{p_ip_j}=1/(M(M+1))$,
and with $V=U_j^\dagger U_i$, $\overline{|\Lambda_{ij}|^2}=\overline{|\langle\psi_t|V|\psi_t\rangle|^2}=(N+|\mathrm{Tr}(U_j^\dagger U_i)|^2)/(N(N+1))$, so $\overline{1-|\Lambda_{ij}|^2}=(N^2-|\mathrm{Tr}(U_j^\dagger U_i)|^2)/(N(N+1))$. Summing over pairs and using $e_{p_0}=2\overline\zeta^{\,p_0}$ gives
Eq.~(\ref{eq:ep0-gcu-exact}).
\end{proof}

\begin{corollary}[GCU-restricted bound]
\label{cor:gcu-bound}
For any GCU gate with $M$ control branches and $N$-dimensional target,
\begin{equation}
e_{p_0}(U) \le \frac{(M-1)N}{(M+1)(N+1)} ,
\label{eq:our-bound}
\end{equation}
with equality iff the target unitaries are pairwise Hilbert--Schmidt orthogonal, $\mathrm{Tr}(U_j^\dagger U_i)=0$ for all $i\ne j$. The same condition identified in \mbox{Section (\ref{sec:3})} for saturating $\zeta_{\max}$.
\end{corollary}
\begin{proof}
Immediate from Eq.~(\ref{eq:ep0-gcu-exact}), each term
$N^2-|\mathrm{Tr}(U_j^\dagger U_i)|^2\le N^2$, with equality iff $\mathrm{Tr}(U_j^\dagger U_i)=0$. There are $\binom{M}{2}$ such pairs.
\end{proof}

Proposition~4 in Ref. \cite{Zanardi_2000} bounds $e_{p_0}(U)$ over \emph{all} $U\in U(d_1d_2)$ (with $d_1\leq d_2$) by $(d_2-d_2/d_1)/(d_2+1)$. Taking $M=d_1$ (using the smaller register as control) in Eq.~(\ref{eq:our-bound}), we get,
\begin{equation}
 e_{p_0}(U) = \frac{d_1}{d_1+1}\cdot\frac{(d_1-1)d_2}{d_1(d_2+1)} \;<\; \frac{d_2-d_2/d_1}{d_2+1}.
\end{equation}

\begin{table}[t]
\centering
\begin{tabular}{cccc}
\hline\hline
$(d_1,d_2)$ ~~~~& GCU bound ~~~~& Reference 5 ~~~~& Optimal \\
\hline
$(2,2)$ ~~~~& $2/9$ ~~~~& $1/3$ ~~~~& $2/9$ \\
$(2,3)$ ~~~~& $1/3$ ~~~~& $3/8$ ~~~~& $1/3$ \\
$(2,4)$ ~~~~& $2/5$ ~~~~& $2/5$ ~~~~& $2/5$ \\
$(3,4)$ ~~~~& $9/20$ ~~~~& $8/15$ ~~~~& $8/15$ \\
\hline\hline
\end{tabular}
\caption{\justifying{Comparison of the GCU-restricted bound Eq.~(\ref{eq:our-bound}) against the reported general bound and optimal values in Reference \cite{Zanardi_2000}. The GCU
construction (mutually orthogonal target unitaries, control on whichever register admits it) exactly reproduces the true optimum for $(2,2)$, $(2,3)$, and $(2,4)$ showing the controlled unitary architecture is not merely a good heuristic but optimal in these cases. For $(3,4)$ a single GCU block falls short of the true
optimum.}}
\label{tab:zanardi-compare}
\end{table}

It means that the upper bound on entangling power derived from the above corollary is tighter than that in Ref. \cite{Zanardi_2000} by a factor of $d_1/(d_1+1)$. Restricting to a single control basis block diagonal GCU gate necessarily costs something relative to the full unitary group. However, taking instead $M=d_2$ (control on the larger register, target dimension $N=d_1$) gives the alternative bound $(d_2-1)d_1/[(d_2+1)(d_1+1)]$, which can exceed the $M=d_1$ bound. The bound $M\le N^2$ follows because the Hilbert-Schmidt product $\langle A,B\rangle:=\mathrm{Tr}(A^\dagger B)$ is an ordinary inner product on the vector space of $N\times N$ complex matrices, and this space has dimension $N^2$. Therefore, any set of mutually Hilbert-Schmidt orthogonal
matrices (in particular, any set of target unitaries $U_i$) cannot exceed $N^2$ in number.
Table~\ref{tab:zanardi-compare} compares the entangling power obtained using Eq.~(\ref{eq:our-bound}) against the general upper bound and reported optimal values in Ref. \cite{Zanardi_2000}. \\

\section{Discussions}
\label{sec:6}
The central advantage of our framework is its computability and predictibility; it enables a circuit designer determine and quantify the entanglement a controlled unitary gate can generate directly from the probability distributions $\{p_i\}$ of the control register's state and the pairwise trace overlaps $\Lambda_{ij}$ of the gate's constituent unitaries, without ever requiring simulation of the output entangled state. This predictive power scales precisely where full state-vector simulation becomes expensive. The example of Toffoli-gate in Section~(\ref{sec:3}) concretely illustrates this fact. Because three of the four target unitaries ($U_{00},U_{01},U_{10}$) coincide with the identity, the gate fails the mutual-orthogonality
criterion, and $\zeta$ flags this immediately. Under a
uniform superposition over all four control branches, the achievable entanglement falls strictly below $\zeta_{\max}$. This shortfall is a property of the uniform control distribution specifically, not of the gate itself. The same Toffoli gate reaches
$\zeta_{\max}$ exactly under the non-uniform distribution $p_{11}=1/2$,
since only the $|11\rangle$ branch carries a nontrivial unitary, as clear from Eq.~(\ref{eq:toffoli-zetar}). Similar analysis can be performed for other entangling gates as well. 

While comparing the newly defined quantity with operator entangling power, we find that GCU-restricted bound of the quantity reproduces the numerically obtained optimal values in \cite{Zanardi_2000} for $(d_1,d_2)=(2,2)$, $(2,3)$, and $(2,4)$. This indicates that the
controlled unitary architecture is, for these dimension pairs, not simply a convenient restriction of the space of bipartite unitaries but the optimal one. The $(3,4)$ case, where a single GCU block falls short of the true optimum, suggests that this correspondence has a limited domain of validity, plausibly tied to whether the optimal unitary is expressible in a single GCU gate or instead requires concatenating controlled operations with alternating control roles~\cite{Zanardi_2000}.

The construction of $\zeta$ suggests several further uses. Because the mutual orthogonality condition is stated pairwise, a gate's shortfall from
$\zeta_{\max}$ can be attributed to specific underperforming branch pairs. Separately, the condition for $\zeta=0$ is an equality case of $|\Lambda_{ij}|\le1$, and the target states saturating it are precisely the eigenstates of the relevant relative unitary, suggesting a route to identifying eigenstates of an unknown unitary from entanglement generation measurements rather than diagonalization. More broadly, since $\zeta$ is built entirely from probability distributions and pairwise overlaps, it is naturally compatible with existing tomography-free estimation techniques such as the Hadamard test and randomized measurement~\cite{purity_from_random_measurements, Elben_purity_randomised_measurements}, offering a resource-efficient predictive approach to analyze entangling gates on real hardware. The main limitation of this approach is its restriction to bipartite settings, which means that it does not capture any entanglement generated between the qubits within the target and control registers themselves.

\section{Conclusion}
\label{sec:7}

We have shown that the entanglement generated by a controlled unitary gate can be completely determined by a  quantity $\zeta$, computable directly from the probability distribution of control register's basis states and the pairwise trace overlaps of the target unitaries, without requiring construction of the output state. For a two qubit controlled unitary gate, the quantity vanishes when the control qubit lacks coherent superposition of its basis states or the target sits in an eigenstate of the conditioning unitary, and reaches its maximum of $1/4$ when the control quibit is in equiprobable superposition of its basis states and the target is prepared suitably. More specifically, a maximally entangled output is achievable for some input \textit{iff} the unitary $U$ is traceless (Theorem 1). This condition also holds good for a generalized controlled unitary gate once expressed through the single relative unitary $U_{\mathrm{rel}}=U_1^\dagger U_0$ that governs its entire entangling behavior (Theorem 2). Extending the analysis to arbitrary control and target register sizes, we have proved that $\zeta$ is upper bounded by $\zeta_{\max}=(2^d-1)/2^{d+1}$ with $d=\min(m,n)$, and that this bound requires populating only $2^d$ control branches rather than a uniform superposition over all $2^m$ of them. Example of the Toffoli gate illustrates this fact, reaching $\zeta_{\max}=1/4$ at $p_{11}=1/2$ despite its four dimensional control register (Theorem 3).

While connecting the proposed quantity to the standard entanglement measures, we find that $\gamma=1-2\zeta$ and $\zeta_r=\zeta/\zeta_{\max}$. Interestingly, for particularly two qubit states, concurrence is
$C=2\sqrt\zeta$,  as in this case $\zeta$ reduces to be the determinant of the reduced density matrix of the control qubit. For a single control qubit we obtained the von Neumann entropy in closed form to arbitrary order. Newton's identities collapse the trace power series to the two term
recursion $\mathrm{Tr}(\rho^k)=\mathrm{Tr}(\rho^{k-1})-\zeta\,\mathrm{Tr}(\rho^{k-2})$, giving $S_2(\rho)=\tfrac{5}{2}\zeta$ at second order and a truncated series that converges monotonically to the exact entropy at every $\zeta\in[0,1/4]$. For a control register of any size $M$, the trace powers $\mathrm{Tr}(\rho^k)$ appearing in the Taylor expansion of von Neumann entropy can always be evaluated directly from Eq.~(\ref{eq:tracepowers}) using only the populations $p_i$ and overlaps $\Lambda_{ij}$, equivalently as $\mathrm{Tr}(A^k)$ for the $M\times M$ matrix $A_{ij}:=\sqrt{p_ip_j}\,\Lambda_{ij}$. This gives $S_N(\rho)$ to any desired order $N$ for a GCU gate of arbitrary register size. The only simplification lost beyond $M=2$ is that $S_N (\rho)$ no longer reduces to a polynomial in $\zeta$ alone, since Newton's identities then require the higher elementary symmetric polynomials $h_3,h_4,\dots$, which are not fixed by $\zeta=h_2$.

Expressing linear entropy as $E(|\Psi\rangle)=2\zeta$ allows us to compute the Haar-averaged entangling power of the controlled unitary architecture in closed form, yielding the bound $e_{p_0}(U)\le(M-1)N/[(M+1)(N+1)]$, which is tighter than the general result \cite{Zanardi_2000} by the factor $M/(M+1)$. This bound reproduces the previous numerically obtained optimal entangling power precisely for $(d_1,d_2)=(2,2)$, $(2,3)$, and $(2,4)$ (values of $2/9$, $1/3$, and $2/5$, respectively), showing that the controlled unitary construction is the optimal entangler rather than a convenient special case for these dimension pairs. The $(3,4)$ case, where this correspondence breaks down, marks the boundary of that optimality. Extension of the present analysis to qudit registers along with experimental estimation via existing tomography-free techniques will be the future step of this work.

\bibliography{bibliography}

@article{Wootters_1998,
   title={Entanglement of Formation of an Arbitrary State of Two Qubits},
   volume={80},
   ISSN={1079-7114},
   url={http://dx.doi.org/10.1103/PhysRevLett.80.2245},
   DOI={10.1103/physrevlett.80.2245},
   number={10},
   journal={Physical Review Letters},
   publisher={American Physical Society (APS)},
   author={Wootters, William K.},
   year={1998},
   month=Mar, pages={2245–2248} }

@article{Wootters_Hill,
   title={Entanglement of a Pair of Quantum Bits},
   volume={78},
   ISSN={1079-7114},
   url={http://dx.doi.org/10.1103/PhysRevLett.78.5022},
   DOI={10.1103/physrevlett.78.5022},
   number={26},
   journal={Physical Review Letters},
   publisher={American Physical Society (APS)},
   author={Hill, Scott and Wootters, William K.},
   year={1997},
   month=June, pages={5022–5025} }

@article{Wootters_Coffman_2000,
   title={Distributed entanglement},
   volume={61},
   ISSN={1094-1622},
   url={http://dx.doi.org/10.1103/PhysRevA.61.052306},
   DOI={10.1103/physreva.61.052306},
   number={5},
   journal={Physical Review A},
   publisher={American Physical Society (APS)},
   author={Coffman, Valerie and Kundu, Joydip and Wootters, William K.},
   year={2000},
   month=Apr }

@article{Zanardi_2000,
   title={Entangling power of quantum evolutions},
   volume={62},
   ISSN={1094-1622},
   url={http://dx.doi.org/10.1103/PhysRevA.62.030301},
   DOI={10.1103/physreva.62.030301},
   number={3},
   journal={Physical Review A},
   publisher={American Physical Society (APS)},
   author={Zanardi, Paolo and Zalka, Christof and Faoro, Lara},
   year={2000},
   month=Aug }

@misc{noisy_gates,
      title={Imperfect Entangling Power of Quantum Gates}, 
      author={Sudipta Mondal and Samir Kumar Hazra and Aditi Sen De},
      year={2025},
      eprint={2401.00295},
      archivePrefix={arXiv},
      primaryClass={quant-ph},
      doi={https://doi.org/10.1103/jl2b-bxfn},
      url={https://arxiv.org/abs/2401.00295}, 
}

@article{Horodecki_2009,
   title={Quantum entanglement},
   volume={81},
   ISSN={1539-0756},
   url={http://dx.doi.org/10.1103/RevModPhys.81.865},
   DOI={10.1103/revmodphys.81.865},
   number={2},
   journal={Reviews of Modern Physics},
   publisher={American Physical Society (APS)},
   author={Horodecki, Ryszard and Horodecki, Pawel and Horodecki, Michal and Horodecki, Karol},
   year={2009},
   month={June}, pages={865–942} }

@book{Nielsen_Chuang_2010, place={Cambridge}, title={Quantum Computation and Quantum Information: 10th Anniversary Edition}, publisher={Cambridge University Press}, author={Nielsen, Michael A. and Chuang, Isaac L.}, year={2010}}

@book{rudin1976,
  author    = {Rudin, Walter},
  title     = {Principles of Mathematical Analysis},
  edition   = {3rd},
  year      = {1976},
  publisher = {McGraw-Hill},
  address   = {New York},
  isbn      = {9780070542358}
}

@article{linearentropy1,
  title = {Mixedness and teleportation},
  author = {Bose, S. and Vedral, V.},
  journal = {Phys. Rev. A},
  volume = {61},
  issue = {4},
  pages = {040101(R)},
  numpages = {2},
  year = {2000},
  month = {Mar},
  publisher = {American Physical Society},
  doi = {10.1103/PhysRevA.61.040101},
  url = {https://link.aps.org/doi/10.1103/PhysRevA.61.040101}
}

@article{linearentropy2,
   title={Maximizing the entanglement of two mixed qubits},
   volume={64},
   ISSN={1094-1622},
   url={http://dx.doi.org/10.1103/PhysRevA.64.030302},
   DOI={10.1103/physreva.64.030302},
   number={3},
   journal={Physical Review A},
   publisher={American Physical Society (APS)},
   author={Munro, W. J. and James, D. F. V. and White, A. G. and Kwiat, P. G.},
   year={2001},
   month=Aug }

@book{newton,
    author = {Macdonald, I G},
    title = {Symmetric Functions and Hall Polynomials},
    publisher = {Oxford University Press},
    year = {1995},
    month = {03},
    isbn = {9780198534891},
    doi = {10.1093/oso/9780198534891.001.0001},
    url = {https://doi.org/10.1093/oso/9780198534891.001.0001},
}

@article{Choi_2023,
   title={Preparing random states and benchmarking with many-body quantum chaos},
   volume={613},
   ISSN={1476-4687},
   url={http://dx.doi.org/10.1038/s41586-022-05442-1},
   DOI={10.1038/s41586-022-05442-1},
   number={7944},
   journal={Nature},
   publisher={Springer Science and Business Media LLC},
   author={Choi, Joonhee and Shaw, Adam L. and Madjarov, Ivaylo S. and Xie, Xin and Finkelstein, Ran and Covey, Jacob P. and Cotler, Jordan S. and Mark, Daniel K. and Huang, Hsin-Yuan and Kale, Anant and Pichler, Hannes and Brandão, Fernando G. S. L. and Choi, Soonwon and Endres, Manuel},
   year={2023},
   month=Jan, pages={468–473} }

@article{Mele_2024haar,
   title={Introduction to Haar Measure Tools in Quantum Information: A Beginner\& amp;apos;s Tutorial},
   volume={8},
   ISSN={2521-327X},
   url={http://dx.doi.org/10.22331/q-2024-05-08-1340},
   DOI={10.22331/q-2024-05-08-1340},
   journal={Quantum},
   publisher={Verein zur Forderung des Open Access Publizierens in den Quantenwissenschaften},
   author={Mele, Antonio Anna},
   year={2024},
   month=May, pages={1340} }

@book{geo_q_state,
  author    = {Bengtsson, Ingemar and {\dots}yczkowski, Karol},
  title     = {Geometry of Quantum States: An Introduction to Quantum Entanglement},
  edition   = {2nd},
  publisher = {Cambridge University Press},
  year      = {2017},
  doi       = {10.1017/9781139207010},
  isbn      = {978-1-107-02625-4}
}

@article{
purity_from_random_measurements,
author = {Tiff Brydges  and Andreas Elben  and Petar Jurcevic  and Benoît Vermersch  and Christine Maier  and Ben P. Lanyon  and Peter Zoller  and Rainer Blatt  and Christian F. Roos },
title = {Probing Rényi entanglement entropy via randomized measurements},
journal = {Science},
volume = {364},
number = {6437},
pages = {260-263},
year = {2019},
doi = {10.1126/science.aau4963},
URL = {https://www.science.org/doi/abs/10.1126/science.aau4963},
eprint = {https://www.science.org/doi/pdf/10.1126/science.aau4963}}

@article{Elben_purity_randomised_measurements,
   title={The randomized measurement toolbox},
   volume={5},
   ISSN={2522-5820},
   url={http://dx.doi.org/10.1038/s42254-022-00535-2},
   DOI={10.1038/s42254-022-00535-2},
   number={1},
   journal={Nature Reviews Physics},
   publisher={Springer Science and Business Media LLC},
   author={Elben, Andreas and Flammia, Steven T. and Huang, Hsin-Yuan and Kueng, Richard and Preskill, John and Vermersch, Benoît and Zoller, Peter},
   year={2022},
   month=Dec, pages={9–24} }

@article{Barenco_1995,
   title={Elementary gates for quantum computation},
   volume={52},
   ISSN={1094-1622},
   url={http://dx.doi.org/10.1103/PhysRevA.52.3457},
   DOI={10.1103/physreva.52.3457},
   number={5},
   journal={Physical Review A},
   publisher={American Physical Society (APS)},
   author={Barenco, Adriano and Bennett, Charles H. and Cleve, Richard and DiVincenzo, David P. and Margolus, Norman and Shor, Peter and Sleator, Tycho and Smolin, John A. and Weinfurter, Harald},
   year={1995},
   month=Nov, pages={3457–3467} }

@article{DiVincenzo_1995,
   title={Two-bit gates are universal for quantum computation},
   volume={51},
   ISSN={1094-1622},
   url={http://dx.doi.org/10.1103/PhysRevA.51.1015},
   DOI={10.1103/physreva.51.1015},
   number={2},
   journal={Physical Review A},
   publisher={American Physical Society (APS)},
   author={DiVincenzo, David P.},
   year={1995},
   month=Feb, pages={1015–1022} }

@book{boas2005mathematical,
  title={Mathematical Methods in the Physical Sciences},
  author={Boas, Mary L.},
  edition={3rd},
  year={2005},
  publisher={John Wiley \& Sons},
  address={Hoboken, NJ}
}

\end{document}